\documentclass[10pt,final,journal]{IEEEtran}
\usepackage{tikz}
\usepackage{amsmath,latexsym,amssymb,amsthm,revsymb,array,bbm,bm}

\newtheorem{definition}{Definition}
\newtheorem{proposition}[definition]{Proposition}
\newtheorem{lemma}[definition]{Lemma}

\newtheorem{theorem}[definition]{Theorem}

\newtheorem{remark}[definition]{Remark}

\definecolor{darkred}{rgb}{0.8,0,0}

\newcommand{\nc}{\newcommand}

\def\mc{\mathcal}

\def\ox{\otimes}

\renewcommand\P[1]{\Pr(#1)} 

\nc{\pe}[3]{\epsilon^{#1}(#2,#3)}
\nc{\mM}[3]{M^{#1}_{#2}(#3)}

\def\NC{\rm{NC}}
\def\SE{\rm{SE}}
\def\NS{\rm{NS}}
\def\b{\ast} 
\def\ppv{\rm{PPV}} 
\def\mSet{\{1,\ldots,M\}} 

\def\tr{\mathrm{Tr}}

\def\mRV{W}
\def\dmRV{\hat{W}}
\def\mV{w}
\def\dmV{\hat{w}}
\def\inA{\mathsf{A}}
\def\outA{\mathsf{B}}
\def\code{\mathcal{Z}}
\def\chan{\mathcal{E}}
\def\src{\mathcal{S}}
\def\is{\!=\!} 

\def\tc{\mathsf{T}}
\def\jointtypes{\mathcal{P}_n(\inA \times \outA)}
\def\jtVar{\tau_{\inA \outA}}
\def\mb{\mathbf}

\def\hg{H} 

\begin{document}
	
\title{A linear program for the finite block length
converse of Polyanskiy-Poor-Verd\'{u}
via non-signalling codes}
\author{William~Matthews
  \thanks{William Matthews (will@northala.net)
is with the Institute for Quantum Computing at the University of Waterloo
and gratefully acknowledges the support of NSERC and QuantumWorks.} }

\maketitle
\begin{abstract}
	Motivated by recent work on
	entanglement-assisted codes for sending
	messages over classical channels,
	the larger, easily characterised class of
	\emph{non-signalling codes} is defined.
	Analysing the optimal performance of
	these codes yields an alternative proof of
	the finite block length converse of
	Polyanskiy, Poor and Verd\'{u},
	and shows that they achieve this converse.
	This provides an explicit formulation of
	the converse as a linear program
	which has some useful features.
	For discrete memoryless channels, it
	is shown that non-signalling codes
	attain the channel capacity with
	zero error probability
	if and only if the dispersion of
	the channel is zero.
\end{abstract}

\section{Introduction}
A key goal of information theory is to quantify the extent to
which reliable communication is possible over a noisy channel.
A code of size $M$ and block length $n$ allows communication
of one of $M$ messages via $n$ uses of the channel.
The fundamental tradeoff between these quantities and
the reliability of communication, is captured by
$\mM{}{\epsilon}{\mc{E}^n}$ - the largest size of code with
error probability $\epsilon$ (for equiprobable messages).
While emphasis is often placed on quantifying
asymptotics of the large $n$ limit (by computing channel capacities
or reliability functions, for example)
but this information isn't necessarily useful if one wishes
to compute a lower bound
on the block length needed for a certain
rate and error probability, for instance.

While actually computing $\mM{}{\epsilon}{\mc{E}^n}$ is intractable,
it is possible to obtain lower (achievability) and upper (converse)
bounds on it from which
the asymptotic quantities derived, but which also
give useful answers for questions concerning finite block lengths.
In their recent paper \cite{PPV}, Polyanskiy, Poor and Verd\'{u} prove a
very general converse bound (the `PPV converse' for the purposes
of this article)
\begin{align}\label{ppv-small}
	\mM{}{\epsilon}{\mc{E}^n} \leq
\mM{\ppv}{\epsilon}{\mc{E}^n},
\end{align}
where $\mM{\ppv}{\epsilon}{\mc{E}^n}$
is given by a maximin optimisation of the reciprocal of the
minimum type II error over a set of hypothesis tests.
They go on to show how many existing converse results can be
easily derived from theirs.

Recent work has shown that it can be advantageous in classical coding
over classical channels for the sender and receiver to share entangled
quantum systems \cite{CLMW,CLMW2,beigi,DSW,LMMOR}. While the capacity
cannot be increased, the number of messages possible for a given
error bound can be. Entanglement assistance can even increase the
zero-error capacity \cite{LMMOR}. This raises questions about the
extent to which entanglement can assist in general.

In an entanglement-assisted code, the
output of the decoder is conditionally independent of the
input to the encoder given the input to the decoder, and vice-versa.
A non-signalling (NS) code is any code with this property,
and $\mM{\NS}{\epsilon}{\mc{E}^n}$ is largest size of NS code
with error probability
$\epsilon$. Any upper bound $\mM{\NS}{\epsilon}{\mc{E}^n}$ clearly
applies to entanglement-assisted codes.

From the elegant proof of the PPV converse \cite{PPV}, it can be seen
that it also applies to non-signalling codes\footnote{My thanks to an
anonymous referee for pointing this out.}, that is
\begin{equation}
	\mM{\NS}{\epsilon}{\mc{E}^n} \leq \mM{\ppv}{\epsilon}{\mc{E}^n}.
\end{equation}
This fact, combined with lower bounds on $\mM{}{\epsilon}{\mc{E}^n}$,
provides quite stringent limits on the advantage
from entanglement assistance. Section \ref{background} precisely
defines the concepts and quantities of interest, and recaps
the proof of the PPV converse.

Section \ref{sec:PNS} analyses performance of non-signalling codes directly,
deriving a linear program for a quantity
$
	\mM{\b}{\epsilon}{\mc{E}^n}
$
whose integer part $\lfloor \mM{\b}{\epsilon}{\mc{E}^n} \rfloor$ is precisely
$
	\mM{\NS}{\epsilon}{\mc{E}^n}.
$
Clearly this quantity is an upper bound on
$\mM{}{\epsilon}{\mc{E}^n}$ and, as mentioned, no larger
than $\mM{\ppv}{\epsilon}{\mc{E}^n}$. Remarkably, it turns out that
$\mM{\b}{\epsilon}{\mc{E}^n}$ is precisely equal to
$\mM{\ppv}{\epsilon}{\mc{E}^n}$. This provides an alternative proof
of the PPV converse (for discrete channels),
which shows that it is \emph{achieved} by NS codes,
and provides primal and dual linear programs (LPs) for it,
which are useful for computing the bound:
The duality theorem for LPs means that any feasible point for the dual LP
gives a valid converse bound, and can allow for certification of optimality.
There is also an operationally intuitive way to use
symmetry of the channel to reduce the size of the linear programs,
from exponential to polynomial in $n$ in the case of discrete memoryless channels (DMCs).

Section \ref{dispersion} shows that DMCs where non-signalling codes
can attain the channel capacity with zero-error, are precisely those
with zero channel dispersion, and thus also admit particularly efficient
classical codes. The final section concludes with some suggested
directions for future research.

\section{Definitions and Background}\label{background}
We consider a single use of a discrete channel with
input alphabet $\inA$ and
output alphabet $\outA$.
The channel input and output
are random variables (RVs) $X$ and $Y$, respectively. 
Our description of the channel use $\mc{E}$
determines the probabilities
$
	\mc{E}(y|x) := \Pr(Y=y|X=x,\mc{E}).
$
A message $\mRV$ is selected from a set of $M$
possible messages $\{1, 2, \ldots, M \}$
by a source $\src$, which determines
the probabilities
$
	\src(w) := \Pr(W=w|\src).
$
A code $\code$ consists of an
encoder, which takes input
$\mRV$ and whose output is the channel input
$X$ from $\inA$,
and a decoder whose input is
the channel output
$Y$ (in $\outA$) and which produces a
decoding $\dmRV$ of the message.
The code $\code$ determines the probabilities
\begin{equation}
	\code(x,\dmV|\mV,y) := \Pr(X=x,\dmRV=\dmV|\mRV=\mV,Y=y,\code).
\end{equation}
An error has occurred if
$\mRV \neq \dmRV$.
\begin{remark}
	When considering $n$ uses of a channel,
the alphabets are $\inA^n$ and $\outA^n$ and the
channel use is $\mc{E}^n$, which gives the conditional
probabilities of output strings
$\mathbf{x} = x_1 \ldots x_n \in \inA^n$
given each input string
$\mathbf{y} = y_1 \ldots y_n \in \outA^n$.
A discrete channel is fully described by specifying
$\bm{\mc{E}} = \{ \mc{E}^n : n \in \mathbb{N} \}$.
A discrete memoryless channel (DMC), is one where
$\mc{E}^n(\mb{y}|\mb{x})
= \mc{E}^{\ox n}(\mb{y}|\mb{x}) := \prod_{i=1}^n \mc{E}(y_i|x_i)$,
for all $n$.
\end{remark}
In a classical code, the encoder and decoder
are uncorrelated, in the sense that
\begin{equation}
	\code(x,\dmV|\mV,y) = F(x|\mV)G(\dmV|y)
\end{equation}
for some conditional probability distributions $F$ and $G$.
This property defines the class $\NC$ of codes with
No Correlation (in the absence of a channel)
between the encoder and the decoder.
\begin{definition}
	$\SE$ (Shared Entanglement) is the class of
	\emph{entanglement-assisted codes} which
	can be implemented by local operations
	of the encoder and
	decoder on quantum systems (with finite Hilbert spaces)
	in a shared entangled state.
	
	A positive operator valued measure (POVM) $L$ for a Hilbert space
	$\mc{H}$, and finite set of outcomes $\mathsf{R}$, assigns
	positive (semidefinite) operators $L(r)$ on $\mc{H}$ to the
	outcomes $r \in \mathsf{R}$ such that
	$\sum_{r \in \mathsf{R}} L(r) = I$,
	where $I$ is the identity operator on $\mc{E}$. A code $\code$
	is in $\SE$ iff
	there exist finite dimensional Hilbert spaces $\mc{H}_A$ and $\mc{H}_B$,
	POVMs $D_{w}$ for $\mc{H}_A$,
	with outcomes in $\inA$ for $w \in \mSet$,
	POVMs $F_{y}$ for $\mc{H}_B$
	with outcomes in $\mSet$
	for $y \in \outA$,
	and a density operator
	$\rho_{AB}$ on $\mc{H}_A\ox \mc{H}_B$, such that
	\[
		\code(x,\hat{w}|w,y) = \tr E_{w}(x) \ox D_{y}(\hat{w}) \rho_{AB}.
	\]
\end{definition}
The class $\SE$ contains the class $\NC$, and is itself contained in
the class of \emph{non-signalling codes}:
\begin{definition}
	A non-signalling (NS) code is any
	one for which the marginal distribution of the output
	of the decoder is conditionally independent of the
	input to the encoder given the input to the decoder,
	and vice-versa. That is,
	for all $x \in \inA,y \in \outA,\mV,\dmV \in \mSet$,
	\begin{align}
		\P{\dmRV\is \dmV|\mRV \is \mV,Y\is y,\code}
		=& \P{\dmRV \is \dmV|Y \is y,\code},\label{NSAtoB0}\\
		\P{X\is x|W\is w,Y\is y,\code}
		=& \P{X\is x|W\is w,\code}.\label{NSBtoA0}
	\end{align}
	These conditions define the class $\NS$ of Non-Signalling
	assisted codes.
\end{definition}
From Bayes' rule and (\ref{NSBtoA0}),
\begin{equation}\label{AtoB-fac}
	\begin{split}
	&\code(x,\dmV|\mV,y)\\
	=& \P{\dmRV\is\dmV|\mRV\is \mV, Y\is y, X\is x,\code}
	p(x|w,\code).
	\end{split}
\end{equation}
where $p(x|w,\code) := \P{X\is x|\mRV\is \mV,\code}$.
This can be interpreted operationally as indicating that if
(\ref{NSBtoA0}) holds, then $\code$ could be implemented by
having the encoder stochastically generate $X$ according
to the value of $\mRV$, and then
send the values of $X$ and $\mRV$ to the decoder (using additional
communication) which would use these, in addition to $Y$, to
determine how to generate $\dmRV$.
Using (\ref{AtoB-fac}) and the fact that
\begin{align}\label{xy-given-w}
	\P{Y\is y,X\is x|W\is w,\code,\chan}
	= \mc{E}(y|x)p(x|w,\code)
\end{align}
it is easy to show that
\begin{equation}\label{cps}
	\begin{split}
		&\P{\dmRV=\dmV,Y=y,X=x|W=w,\chan,\code,\src}\\
		=&\code(x,\hat{w}|w,y)\chan(y|x).
	\end{split}
\end{equation}
\begin{proposition}\label{ns-fdbk}
	The conditional probabilities $(\ref{cps})$
	are clearly non-negative.
	To form a valid conditional distribution, they must also satisfy
	\begin{align}\label{vpmf}
		\forall w: \sum_{\hat{w},x,y} \code(x,\hat{w}|w,y)\mc{E}(y|x)
		= 1.
	\end{align}
	This is true for all channels $\mc{E}$
	if and only if
	$\code$ is non-signalling from the receiver to the sender
	(this is the condition expressed by (\ref{NSBtoA0})).
\end{proposition}
\begin{proof}
	(\ref{vpmf}) is a straightforward consequence of (\ref{NSBtoA0}) via
	(\ref{AtoB-fac}). For the other direction,
	if $\code$ is signalling from Bob to Alice then
	there exist $w' \in \mSet$, $x' \in \inA$ and $y_0, y_1 \in \outA$
	such that
	$\sum_{\hat{w}}\code(x',\hat{w}|w',y_0)
	> \sum_{\hat{w}}\code(x',\hat{w}|w',y_1)$.
	Choosing the channel $\mc{E}$ with
	$
		\mc{E}(y_0|x') = 1
	$
	and, for all $x \neq x'$,
	$
		\mc{E}(y_1|x) = 1,
	$
	\begin{equation}\sum_{x,\hat{w}} \mc{E}(y_0|x)\code(x,\hat{w}|w',y_0)
		> \sum_{x,\hat{w}} \mc{E}(y_0|x)\code(x,\hat{w}|w',y_1)
	\end{equation}
	Since $\forall x: \mc{E}(y_0|x)
	= 1 - \mc{E}(y_1|x)$,
	this implies that
	\begin{align}
		\sum_{\hat{w},x,y} \code(x,\hat{w}|w',y) \mc{E}(y|x) > 1.
	\end{align}
\end{proof}
For the rest of this paper, the source is taken to be $\src_M$,
which assigns equal probability to each message: $\forall w: S_M(w) = 1/M$.

\begin{definition}\label{code-params}
For channel $\mc{E}$, the minimum average probability of error which
can be achieved by a code in class $\Omega$ is
\begin{equation*}
	\pe{\Omega}{M}{\mc{E}}
	:= \min \{ \Pr(\mRV\!\neq\!\dmRV|\code,\mc{E},\src_M)
	: \code \text{ in } \Omega\}
\end{equation*}
and the largest local code with error no larger than $\epsilon$ has
size
\begin{equation*}
	\mM{\Omega}{\epsilon}{\mc{E}}
	:= \max \{ M :
	\Pr(\mRV\!\neq\!\dmRV|\code,\mc{E},\src_M) \leq \epsilon,
	\code \text{ in } \Omega\}.
\end{equation*}
When the superscript $\Omega$ is omitted, it is intended that $\Omega = \NC$.
\end{definition}
\begin{remark}
	By the inclusions of the classes of codes,
	\begin{equation}
		\pe{}{M}{\mc{E}}
		\geq \pe{\SE}{M}{\mc{E}}
		\geq \pe{\NS}{M}{\mc{E}},
	\end{equation}
	and
	\begin{equation}
		\mM{}{\epsilon}{\mc{E}}
		\leq \mM{\SE}{\epsilon}{\mc{E}}
		\leq \mM{\NS}{\epsilon}{\mc{E}}.
	\end{equation}
\end{remark}
\begin{remark}\label{useless}
	Note that if $\mc{E}(y|x) = q(y)$ then, using (\ref{NSAtoB0}),
\begin{equation}
	\begin{split}
	&\pe{\NS}{M}{\mc{E}}
	= 1- \frac{1}{M}\sum_{w,x,y}\code(x,w|w,y)q(y)\\
	&= 1- \frac{1}{M}\sum_{w,y}\P{\hat{W}=w|Y=y}q(y) = 1 - 1/M.
	\end{split}
\end{equation}
\end{remark}
$\pe{}{M}{\mc{E}}$ or $\mM{}{\epsilon}{\mc{E}}$ are in general both hard to
compute and to
analyse. This motivates the desire for bounds on these quantities which
are more amenable to computation and/or analysis. Many previously
established converse results can derived from the following result
of Polyanskiy, Poor and Verd\'{u}:

\def\PDs{\mathcal{P}}
\def\gSet{\mathsf{C}}

\begin{definition}
	For a finite set $\gSet$, let $\PDs(\gSet)$ denote the set of
	probability distributions on $\gSet$.
	Given distributions $P^{(0)}, P^{(1)} \in \PDs(\gSet)$,
	(and identifying $P^{(0)}$ with the null hypothesis)
	let
	$\beta_{1-\epsilon}(P^{(0)},P^{(1)})$
	denote the
	minimum type II error
	$\sum_{r \in \gSet} T_r P^{(1)}(r)$
	that can be achieved by statistical tests
	$T$ which give a type I error no greater than $\epsilon$, i.e.
	$\sum_{r \in \mathsf{C}} T_r P^{(0)}(r) \geq 1 - \epsilon$.
\end{definition}
\begin{theorem}[PPV converse - Theorem 27 of \cite{PPV}]\label{PPV}
	The number of messages which can be transmitted by an NS code with
	probability no greater than $\epsilon$ obeys $\mM{\NS}{\epsilon}{\mc{E}}
	\leq \mM{\ppv}{\epsilon}{\mc{E}}$ where
	\begin{align}
		\mM{\ppv}{\epsilon}{\mc{E}}
		:= 
		\displaystyle\max_{P_{X}\in \PDs(\inA)} \min_{Q_{Y}\in \PDs(\outA)}
		\frac{1}{
		\beta_{1-\epsilon}(P_{XY},P_{X}\times Q_{Y})
		}
	\end{align}
	with $P_{XY}(x,y) := P_X(x)\mc{E}(y|x)$.
\end{theorem}
\begin{proof}
	Define two hypotheses $H_0$ and $H_1$
	to explain the data $X = x, Y = y$:
	In both $\code$ is an NS code of size $M$,
	but in $H_i$ the channel is $\mc{E}_i$.
	Let $\epsilon_i$ denote the error probability
	attained by the code if the channel is $\mc{E}_i$:
	\begin{equation}
		1-\epsilon_i = \P{W=\hat{W}|\code,\chan_i,\src_M}
		= \sum_{x,y} T_{xy} P^{(i)}_{XY}(x,y)
	\end{equation}
	where
	\begin{align}
	P^{(i)}_{XY}(x,y)
	:=& \P{X=x,Y=y|\code,\chan_i,\src_M}\\
	=& \chan_i(y|x) \P{X=x|\code,\src_M}
	\end{align}
	and
	\begin{align}
	T_{xy} :=& \P{W=\hat{W}|X=x,Y=y,\code,\src_{M}}
	\end{align}
	which, using (\ref{xy-given-w}) and (\ref{cps}), is
	\begin{equation}
	T_{xy} = \sum_{w=1}^M \frac{\code(x,w|w,y)}{M p(x|w,\code)}.
	\end{equation}
	The direct part of Proposition \ref{ns-fdbk} guarantees that this
	is a valid statistical test, which proves that
	\begin{equation}
		\beta_{1-\epsilon_0}(P^{(0)}_{XY},P^{(1)}_{XY}) \leq 1 - \epsilon_1.
	\end{equation}
	Setting $\chan_0 = \chan$ and $\chan_1(y|x) = Q_{Y}(y)$
	(see Remark \ref{useless}) shows that, for
	any NS code with $\Pr(X=x|\code,\src_M) = P_{X}(x)$,
	$\epsilon$ and $M$ must satisfy
	\begin{align}
		\max_{Q_{Y}} \beta_{1-\epsilon}(P_{XY},P_{X}\times Q_{Y})
		\leq \frac{1}{M},
	\end{align}
	and therefore,
	\begin{align}
		\min_{P_{X}} \max_{Q_{Y}}
		\beta_{1-\epsilon}(P_{XY},P_{X}\times Q_{Y})
		\leq \frac{1}{\mM{\NS}{\epsilon}{\mc{E}}}.
	\end{align}
\end{proof}
\begin{definition}
	For codes in class $\Omega$,
	the $\epsilon$-error capacity of $\bm{\mc{E}}$
	\begin{align}
		C_{\epsilon}^{\Omega}(\bm{\mc{E}})
		:= \lim_{n\to\infty}
		\frac{1}{n} \log \mM{\Omega}{\epsilon}{\mc{E}^n}
	\end{align}
	and the capacity is
	\begin{align}
		C^{\Omega}(\bm{\mc{E}})
		:= \lim_{\epsilon\to 0} C_{\epsilon}^{\Omega}(\bm{\mc{E}}).
	\end{align}
\end{definition}
The fact that the PPV converse applies to NS codes has some
immediate consequences:
\begin{remark}
	Since the information spectrum converse that Verd\'{u} and Han
	use to derive their general formula for channel capacity
	\cite{VerduHan} can be
	derived from the PPV converse, this formula also gives the
	capacity for NS codes.
\end{remark}
\begin{remark}
	For DMCs, a converse derived from the PPV converse
	and an achievability bound for classical codes,
	can be used to prove \cite{PPV}
	the result of Strassen \cite{strassen1962asymptotische},
	\begin{equation}\label{strassen}
		\log \mM{}{\epsilon}{\mc{E}^{\ox n}}
		= n C - \sqrt{n V} Q^{-1}(\epsilon) + O(\log n),
	\end{equation}
	where $C$ is the channel capacity,
	$V$ is the \emph{channel dispersion} (see Section \ref{dispersion})
	and $Q(x) := (2 \pi)^{-1/2} \int_{x}^{\infty} e^{-t^{2}/2} dt$.
	
	Since the PPV converse also applies to NS codes,
	(\ref{strassen}) applies to these too,
	and the difference in the rates achieved by classical and
	NS codes (for fixed $\epsilon$) is only of
	order $O(\log n)/n$.
\end{remark}
\section{The performance of non-signalling codes}\label{sec:PNS}
The optimisation over codes that yields $\pe{\NS}{M}{\mc{E}}$ in
Definition (\ref{code-params}) is already a linear program (LP):
The variable is the code $\code$
(considered as a $|\inA||\outA|M^2$ dimensional real vector),
the objective function $\Pr(\mRV\!\neq\!\dmRV|\code,\mc{E},\src_M)$ is
\begin{equation}
	1 - \frac{1}{M}\sum_{w,x,y} \mc{E}(y|x)\code(x,w|w,y),
\end{equation}
and the constraints are simply the linear equalities comprising the
non-signalling conditions (\ref{NSAtoB0}) and (\ref{NSBtoA0}),
in addition to the non-negativity and normalisation of $\code$.

\begin{figure}[ht]
	\begin{tikzpicture}[scale=1.2]
		\node (input) at (-3.1,1) {$\mRV$};
		\node (aperm) at (-2.4,1) {};
		\node (aplab) at (-1.7,1.3) {$\pi(W)$};
		\node (abox) at (-0.9,1) {};
		\node[rectangle, draw, thick, fill=white, inner sep=3mm]
		(chan) at (0,0) {$\mc{E}$};
		\node (bbox) at (1,-1) {};
		\node (bplab) at (1.8,-.7) {$\pi(\hat{W})$};
		\node (bperm) at (2.7,-1) {};
		\node (output) at (3.6,-1) {$\dmRV$};
		\draw[draw=black, style=double]
		(input) -- (abox) -- (-.5,1) -- (chan);
		\draw[draw=black, style=double]
		(chan) -- (.5,-1) -- (bbox) -- (output);
		\node[rectangle, draw, thick, fill=black!20, inner sep=3mm]
		(abox2) at (abox) {$e$
		};
		\node[rectangle, draw, thick, fill=white, inner sep=2mm]
		(aperm2) at (aperm) {$\pi$
		};
		\node[rectangle, draw, thick, fill=black!20, inner sep=3mm]
		(bbox2) at (bbox) {$d$
		};
		\node[rectangle, draw, thick, fill=white, inner sep=2mm]
		(bperm2) at (bperm) {$\pi^{-1}$
		};
		\node[above] at (chan.north) {$X$};
		\node[below] at (chan.south) {$Y$};
	\end{tikzpicture}
  \centering
  \caption{Operational interpretation of the code $\bar{\code}$
which results from the symmetrisation (\ref{sym}) of a non-signalling code.
The boxes marked `$e$' and `$d$' are the encoder and decoder for the original
non-signalling code $\code$.
The permutations are coordinated by a shared random variable $\pi$
drawn uniformly at random from the symmetric group on $\mSet$.}
\label{fig:sym}
\end{figure}

If $\code$ is an NS code, then let
\begin{equation}\label{sym}
	\bar{\code}(x,\dmV|\mV,y)
	= \frac{1}{|G|} \sum_{\pi \in G} \code(x,\pi(\dmV)|\pi(\mV),y)
\end{equation}
where $G$ is the symmetric group on $\mSet$, $\pi(\mV)$ denotes the
action of a permutation in $G$ on $\mV \in \mSet$.
This symmetrized code $\bar{\code}(x,\dmV|\mV,y)$ has an operational
interpretation given in Fig \ref{fig:sym}, from which it is clear that it
is also non-signalling and since
\begin{align}
	&\P{W=\hat{W}|\bar{\code},\chan,\src_M}\\
	=&\frac{1}{|G|} \sum_{x,y} \sum_{\pi \in G} \sum_w
	\mc{E}(y|x)\code(x,\pi(\mV)|\pi(\mV),y)\\
	=&\frac{1}{|G|} \sum_{\pi \in G} \sum_{x,y} \sum_w
	\mc{E}(y|x)\code(x,\mV|\mV,y)\\
	=& \P{W=\hat{W}|\code,\chan,\src_M},
\end{align}
the optimisation over NS codes for $\pe{\NS}{M}{\chan}$,
can be restricted to symmetrized codes.
These are precisely those codes with the form
\begin{equation}
	\code(x,\dmV|\mV,y) = \begin{cases}
							R_{xy} &\text{ if } \dmV = \mV,\\
							Q_{xy} &\text{ if } \dmV \neq \mV.
						\end{cases}
\end{equation}
In these terms, the non-signalling condition (\ref{NSBtoA0}) is equivalent
to saying that there exists $p: \inA \to \mathbb{R}$ such that
$
	R_{xy} + (M-1)Q_{xy} = p(x),
$
and so
\begin{equation}
	\code(x,\dmV|\mV,y) = \begin{cases}
							R_{xy} &\text{ if } \dmV = \mV,\\
							(p(x) - R_{xy})/(M-1) &\text{ if }
							\dmV \neq \mV.
						\end{cases}
\end{equation}
With this simplification, the conditional
probabilities in $\code$ are non-negative iff $R_{xy} \geq 0$ and
$p(x) \geq R_{xy}$ for all $x,y$, and the normalisation condition
$\forall {w,y}: \sum_{x,\hat{w}} \code(x,\hat{w}|w,y) = 1$
is equivalent to $\sum_{x} {p(x)} = 1$.
The condition (\ref{NSAtoB0}) of no signalling from encoder to decoder is
$\forall y: \sum_{x} {R_{xy}} = \sum_{x}(p(x) - R_{xy})/(M-1)$
which, in light of the normalisation condition, is equivalent to
\begin{equation}\label{Rineq}
	\forall y: \sum_{x} R_{xy} = 1/M.
\end{equation}
Given a feasible point with, 
$\sum_{x\in \inA} R_{xy'} < 1/M$ for some $y' \in \outA$,
for any $\lambda \in [0,1]$,
\begin{equation}
	R'_{xy} = \begin{cases}
	(1-\lambda)R_{xy} + \lambda p(x) &\text{ for } y = y',\\
	R_{xy} &\text{ otherwise}
	\end{cases}
\end{equation}
is also feasible and there must exist $\lambda$ s.t.
$\sum_{x\in \inA} R_{xy'} = 1/M$.
Since $R'_{xy} \geq R_{xy}$ for all $x,y$ this can only be an
improvement on the original point, so (\ref{Rineq}) can be changed
to an inequality, to obtain
\begin{proposition}\label{errorPrimal}
\begin{align}
	1 - \pe{\NS}{M}{\mc{E}} =
	&\max \sum_{x \in \inA}\sum_{y \in \outA} \mc{E}(y|x) R_{xy}
	\label{pObj}\\
	&\text{subject to}\\
	\forall y \in \outA:&\sum_{x\in \inA} R_{xy} \leq 1/M,\label{RMeq}\\
	\forall x \in \inA, y \in \outA :&p(x) \geq R_{xy},\label{Ru}\\
	&\sum_{x \in \inA} p(x) = 1,\label{RuNorm}\\
	\forall x \in \inA, y \in \outA :&R_{xy} \geq 0, p_{x} \geq 0\label{pos}.
\end{align}
\end{proposition}

Introducing Lagrange multipliers $D_{xy}$, $z_{y}$, $\alpha$ for the
constraints (\ref{Ru}), (\ref{RMeq}), (\ref{RuNorm}) respectively,
the Lagrangian function is
\begin{align}
	&\sum_{x\in \inA}\sum_{y\in \outA} \mc{E}(y|x) R_{xy} 
	+\sum_{x \in \inA}\sum_{x \in \outA} D_{xy} (p(x) - R_{xy})\\
	+& \sum_{y \in \outA}z_{y}
	\left(\frac{1}{M} - \sum_{x \in \inA} R_{xy}\right)
	+ \alpha\left(1 - \sum_{x\in \inA} p(x)\right)\\
	=&\sum_{x \in \inA}\sum_{x \in \outA}
	R_{xy} \left(\mc{E}(y|x) - D_{xy} - z_y\right)\\
	+&\sum_{x \in \inA} p(x) \left(\sum_{y \in \outA} D_{xy} - \alpha\right)
	+\alpha + \frac{1}{M} \sum_{y \in \outA} z_y.
\end{align}
Taking the supremum over non-negative $R$ and $u$ and restricting the
multipliers to the region where it is finite yields the dual LP,
whose solution is equal to that of the primal LP by the strong duality
theorem for linear programming:
\begin{align}
	\pe{\NS}{M}{\mc{E}} =
	&\max \left( 1 - \alpha - \frac{1}{M} \sum_{y \in \outA} z_y \right)\\
	&\text{{\rm subject to}}\\
	\forall x \in \inA, y \in \outA :&\mc{E}(y|x) \leq D_{xy} + z_y,\\
	\forall x \in \inA:&\sum_{y \in \outA} D_{xy} \leq \alpha,\\
	\forall x \in \inA, y \in \outA :&D_{xy} \geq 0.
\end{align}
Fixing $z$, one should pick $D_{xy} = \max \{\mc{E}(y|x) - z_y, 0 \}$ and
$\alpha = \max_{x \in \inA} \sum_{y \in \outA} D_{xy}$ so that the objective
function is
\begin{align}
&1 - \max_{x\in \inA} \sum_{y \in \outA} \max \{\mc{E}(y|x) - z_y, 0 \}
-\frac{1}{M} \sum_{y \in \outA} z_y \\
=& \min_{x \in \inA} \sum_{y \in \outA}
(\mc{E}(y|x) -\max \{\mc{E}(y|x) - z_y, 0 \} )
-\frac{1}{M} \sum_{y \in \outA} z_y \\
=& \min_{x \in \inA} \sum_{y \in \outA} \min \{\mc{E}(y|x), z_y \}
- \frac{1}{M} \sum_{y \in \outA} z_y.
\end{align}
It remains to maximise over $z$:
\begin{proposition}\label{errorDualSimp}
	The minimum error probability which can be attained by an NS code is
	\begin{equation*}
	\pe{\NS}{M}{\mc{E}} = \max_{z} \min_{x \in \inA} \sum_{y\in \outA}
	( \min\left\{ z_y, \mc{E}(y|x) \right\} - z_y / M ).
	\end{equation*}
\end{proposition}
Allowing $M$ to take on real values in Proposition \ref{errorPrimal}
and defining $\mu := 1/M$, it is evident that $\pe{\NS}{M}{\mc{E}}$ is a
piece-wise linear, non-increasing, concave function of $\mu$ for
$\mu \in [0,1]$. What's more, this can be inverted to obtain a
linear program which gives the smallest value of $1/M$ such that there
exists an NS code of size $\lfloor M \rfloor$ with
error probability $\epsilon$ for $\mc{E}$. That is,
$
	\mM{\NS}{\epsilon}{\mc{E}} = \lfloor \mM{\b}{\epsilon}{\mc{E}} \rfloor
$
where
\begin{align}
	\mM{\b}{\epsilon}{\mc{E}}^{-1} =& \min \mu,\\
	&\text{subject to}\\
	\forall y \in \outA:&\sum_{x \in \inA} R_{xy} \leq \mu, \label{c2}\\
	&\sum_{x\in \inA}\sum_{y \in \outA} \mc{E}(y|x) R_{xy} 
	\geq 1 - \epsilon,\\
	&\text{and the constraints (\ref{Ru}-\ref{pos})}.
\end{align}
At this point, it is quite straightforward to show the claimed equivalence
to the PPV converse:
\begin{proposition}
	\begin{equation}
		\mM{\b}{\epsilon}{\mc{E}} = \mM{\ppv}{\epsilon}{\mc{E}}
	\end{equation}
\end{proposition}
\begin{proof}
Writing out the optimisation that determines the PPV converse
(Theorem \ref{PPV}) explicitly (with the shorthands
$p(x) := P_X(x)$, $q(x) := Q_Y(y)$),
it is clear that the function being optimised
is bilinear in $T$ and $q$, both of which are constrained to finite
dimensional polytopes. Using von Neumann's minimax theorem \cite{minimax},
\begin{align}
	\mM{\ppv}{\epsilon}{\mc{E}}^{-1}
	=& \min_p \max_q \min_T
	\sum_{x \in \inA}\sum_{y\in \outA}T_{xy} p(x) q(y)\\
	=& \min_p \min_T \max_q
	\sum_{x \in \inA}\sum_{y\in \outA}T_{xy} p(x) q(y)\\
	=& \min_p \min_T \max_{y\in \outA} \sum_{x \in \inA}T_{xy} p(x)\\
	&\text{subject to}\\
	& \sum_{x \in \inA}\sum_{y\in \outA} \mc{E}(y|x) p_{x} T_{xy}
	\geq 1 - \epsilon,\label{PPVc1}\\
	&\sum_x p(x) = 1, \sum_y q(y) = 1,\label{PPVc2}\\
	\forall x \in \inA, y \in \outA:& 0 \leq T_{xy} \leq 1,
	\label{PPVc3}\\
	\forall x, y:&p(x) \geq 0, q(y) \geq 0.\label{PPVc4}
\end{align}
Writing $R_{xy} = p(x) T_{xy}$, this linear program is equivalent to
\begin{align}
	&\min \mu\\
	&\text{subject to}\\
	\forall y \in \outA:&\sum_{x \in \inA} R_{xy} \leq \mu,\label{cA1}\\
	& \sum_{x \in \inA}\sum_{y\in \outA} \mc{E}(y|x) R_{xy}
	\geq 1 - \epsilon,\\
	& \sum_{x\in \inA} p(x) = 1,\label{norm1}\\
	\forall x\in \inA,y\in \outA:& 0 \leq R_{xy} \leq p(x)\label{cA2},
\end{align}
which is exactly the primal LP for $\mM{\b}{\epsilon}{\mc{E}}^{-1}$.
\end{proof}

Since the maximisation of $\frac{1}{\mu}$ under
the constraints (\ref{cA1}-\ref{cA2}), which
yields
$M^{\b}_P(\mc{E})$
directly, is a linear-fractional program
\cite{BoydVandenberghe}, the Charnes-Cooper transformation
\cite{CharnesCooper}
\begin{align}
	F_{xy} := R_{xy}/\mu,~
	v_x := p(x)/\mu,~
	t := 1/\mu,
\end{align}
can be used to transform it into a linear program for
$\mM{\b}{\epsilon}{\mc{E}}$, from which $t$ can be eliminated
by using the transformed version of (\ref{norm1}),
$\sum_{x\in \inA}v_x = t$, to obtain 
\begin{theorem}
	\label{messagesPrimal}  
	$\mM{\NS}{\epsilon}{\mc{E}} = \lfloor \mM{\b}{\epsilon}{\mc{E}} \rfloor$,
	where
\begin{align}
	\mM{\b}{\epsilon}{\mc{E}} &= \max \sum_{x\in \inA} v_x,\\
	&\text{{\rm subject to}}\\
	\forall x \in \inA, y \in \outA :&F_{xy} \leq v_x \label{MpXY},\\
	\forall y \in \outA:&\sum_{x\in \inA} F_{xy} \leq 1 \label{MpY},\\
	&\sum_{x \in \inA}\sum_{y\in \outA} \mc{E}(y|x) F_{xy}
	\geq (1-\epsilon)\sum_{x\in \inA} v_x \label{MpP},\\
	\forall x \in \inA, y \in \outA :&F_{xy} \geq 0, v_{x} \geq 0.
\end{align}
\end{theorem}

Since the main goal is to obtain \emph{upper} bounds on $M$,
the dual of this linear program is more useful.
Introducing Lagrange multipliers
$V_{xy}$, $c_y$ and $\xi$ for the constraints (\ref{MpXY}), (\ref{MpY}) and
(\ref{MpP}) respectively,
taking the infimum of the resulting Lagrangian over non-negative $F$ and $v$,
and restricting the multipliers to the finite region gives us
the dual program:
\begin{theorem}
	\label{messagesDual}
	$\mM{\NS}{\epsilon}{\mc{E}} = \lfloor \mM{\b}{\epsilon}{\mc{E}} \rfloor$,
	where
\begin{align}
	\mM{\b}{\epsilon}{\mc{E}} &= \min \sum_{y \in \outA} c_y,\\
	&\text{{\rm subject to}}\\
	\forall x \in \inA, y \in \outA :
	&V_{xy} + c_y\geq \zeta \mc{E}(y|x) \label{MdXY},\\
	\forall x \in \inA:&\sum_{y \in \outA} V_{xy}
	\leq (1-\epsilon)\zeta - 1 \label{MdX},\\
	\forall x \in \inA, y \in \outA :&V_{xy} \geq 0, c_{y} \geq 0.
\end{align}
\end{theorem}
At any feasible point of this dual LP, the value of the
objective function is an upper bound on $\mM{\NS}{\epsilon}{\mc{E}}$.
\subsection{The zero-error case.}
In \cite{CLMW2}, it was shown that $\mM{\NS}{0}{\mc{E}}$ is given by a
linear program which is determined by a combinatorial object associated
with $\mc{E}$, namely its \emph{hypergraph}. This subsection recovers that
result as a special case of the results developed here.
First, some definitions:
	The \emph{hypergraph} $\hg(\mc{E})$ of $\mc{E}$ has vertex set
	$V(H) = \inA$ and hyperedges
	\begin{equation}
		E(H(\mc{E}))
		:= \{ e_{y} := \{ x : \mc{E}(y|x) > 0 \} : \forall y \in Y \}
	\end{equation}
	capturing the equivocation of each output symbol $y \in \outA$.
	(Note that since the set of hyperedges is defined by its members,
	these being subsets of $\inA$, the number of hyperedges may be less
	than the number of output symbols.)
	\label{defi:fractional}
	A \emph{fractional packing} of a hypergraph $H$
	is an assignment of non-negative weights
	$v(x) \leq 1$ to all vertices $x \in V(H)$ such that
  \begin{equation}
    \forall e \in E(H): \quad \sum_{x \in e} v_x \leq 1.
  \end{equation}
	A \emph{fractional covering} of a hypergraph $H$
	is an assignment of non-negative weights
	$c(e) \leq 1$ to all hyperedges
	$e \in E(H)$ such that
	\begin{equation}
	\forall\ x\in \inA: \quad \sum_{e \ni x} c_e \geq 1.
  	\end{equation}
  (Restricting the weights to $\{0,1\}$ recovers
the combinatorial notions of packing and covering.)

The \emph{fractional packing number} $\alpha^*(H)$ is the maximum total
weight allowed in a fractional packing of $H$ and the
\emph{fractional covering number} $\omega^*(H)$ is the minimum total
weight required for a fractional covering of $H$. These are clearly
dual linear programs, which for a channel hypergraph $H(\mc{E})$
have the formulation
\begin{align*}
	\alpha^*(H(\mc{E})) =
	\max \bigg\{ & \sum_{x \in \inA} v_x :
	\forall x\in \inA, v(x) \geq 0,\\
    & \sum_{y \in \outA} \lceil \mc{E}(y|x) \rceil v_x \leq 1 \bigg\},\\
    \omega^*(H(\mc{E})) =
	\min \bigg\{ & \sum_{y \in \outA} c_y :
	\forall y\in \outA, c_y \geq 0,\\
    & \sum_{x \in \inA} \lceil \mc{E}(y|x) \rceil c_y \geq 1 \bigg\},
\end{align*}
(note that $\lceil \mc{E}(y|x) \rceil$ is $0$ if $\mc{E}(y|x) = 0$
and is otherwise $1$.)

In \cite{CLMW2} it was shown that
$\mM{\NS}{0}{\mc{E}} = \lfloor \alpha^{\ast}(H(\mc{E}))\rfloor$.
Given Theorem \ref{messagesPrimal}, this is equivalent to
\begin{proposition}\label{zero-error}
	\begin{equation}
		\mM{\b}{0}{\mc{E}} = \omega^{\ast}(H(\mc{E}))
		= \alpha^{\ast}(H(\mc{E})).
	\end{equation}
\end{proposition}
\begin{proof}
	In the primal LP for $\mM{\b}{0}{\mc{E}}$ (Theorem \ref{messagesPrimal}),
	let $v_x$ be any fractional packing of $H(\mc{E})$, and let
	\begin{equation}
		F_{xy} =
		\begin{cases}
			v_x &\text{ if } \mc{E}(y|x) > 0,\\
			0 &\text{ otherwise. }
		\end{cases}
	\end{equation}
	Now, the constraints (\ref{MpXY}) are trivially satisfied and the
	constraint (\ref{MpP}) is satisfied because
	$\sum_{x \in \inA}\sum_{y\in \outA} \mc{E}(y|x) F_{xy}
	= \sum_{x \in \inA}\sum_{y\in \outA} \mc{E}(y|x) v_x
	= \sum_{x\in \inA} v_x$.
	For all $y \in \outA$,
	$\sum_{x\in \inA} F_{xy} = \sum_{x: \mc{E}(y|x) > 0} v_x$
	which is less than or equal to one because $v_x$ is a fractional packing,
	so the constraints (\ref{MpY}) are satisfied.
	Therefore, \begin{align}\label{fpacking}
		\alpha^{\ast}(H(\mc{E})) \leq \mM{\b}{0}{\mc{E}}.
	\end{align}

	In the dual LP for $\mM{\b}{0}{\mc{E}}$ (Theorem \ref{messagesDual}),
	let $c_y$ be any fractional covering of $H(\mc{E})$,
	choose the smallest $\zeta$ such that
	$\forall x,y: \zeta \mc{E}(y|x) \geq c_y$, and let
	$V_{xy} = \max\{ 0, \zeta \mc{E}(y|x) - c_y \}$.
	Clearly the constraints (\ref{MdXY}) are satisfied, and
	for all $x \in \inA$,
	\begin{align}
		\sum_{y \in \outA} V_{xy}
		= \sum_{ y : \mc{E}(y|x) > 0 } (\zeta \mc{E}(y|x) - c_y) 
		\leq \zeta - 1,
	\end{align} as required for (\ref{MdX}). Therefore,
	\begin{align}\label{fcovering}
		\mM{\b}{0}{\mc{E}} \leq \omega^{\ast}(H(\mc{E})).
	\end{align}
	Since $\omega^{\ast}(H(\mc{E})) = \alpha^{\ast}(H(\mc{E}))$, the result
	follows.
\end{proof}
\subsection{Taking advantage of symmetry}
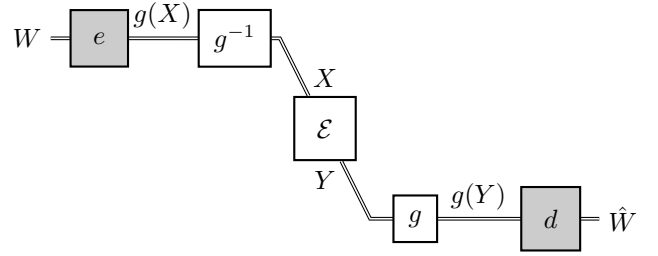
\begin{figure}[ht]
	\begin{tikzpicture}[scale=1.2]
		\node (input) at (-3.3,1) {$\mRV$};
		\node (aperm) at (-1,1) {};
		\node (aplab) at (-1.8,1.25) {$g(X)$};
		\node (abox) at (-2.5,1) {};
		\node[rectangle, draw, thick, fill=white, inner sep=3mm]
		(chan) at (0,0) {$\mc{E}$};
		
		\node (aplab) at (1.7,-.75) {$g(Y)$};
		\node (bperm) at (1,-1) {};
		\node (bbox) at (2.5,-1) {};
		\node (output) at (3.3,-1) {$\dmRV$};
		
		\draw[draw=black, style=double]
		(input) -- (abox) -- (-.5,1) -- (chan);
		\draw[draw=black, style=double]
		(chan) -- (.5,-1) -- (bbox) -- (output);
		\node[rectangle, draw, thick, fill=black!20, inner sep=3mm]
		(abox2) at (abox) {$e$
		};
		\node[rectangle, draw, thick, fill=white, inner sep=2mm]
		(aperm2) at (aperm) {$g^{-1}$
		};
		\node[rectangle, draw, thick, fill=black!20, inner sep=3mm]
		(bbox2) at (bbox) {$d$
		};
		\node[rectangle, draw, thick, fill=white, inner sep=2mm]
		(bperm2) at (bperm) {$g$
		};
		\node[above] at (chan.north) {$X$};
		\node[below] at (chan.south) {$Y$};
	\end{tikzpicture}
  \centering
  \caption{Operational interpretation of the code $\bar{\code}$
which results from the symmetrisation (\ref{sym2}) of a non-signalling code.
The boxes marked `$e$' and `$d$' are the encoder and decoder for the original
non-signalling code $\code$.
The transformations of the channel input and output
are coordinated by a shared random variable $g$
drawn uniformly at random from the group $G$.}
\label{fig:sym2}
\end{figure}

\def\ga{\circ}
Let $G$ be a group with an action on the input alphabet $\inA$
and on the output alphabet $\outA$ (inducing a joint
action on $\inA \times \outA$), such that
\begin{equation}
	\forall g \in G : \mc{E}(g \ga y | g \ga x) = \mc{E}(y|x).
\end{equation}
For any non-signalling code $\code$ define the code
\begin{equation}\label{sym2}
\bar{\code}(x,\dmV|\mV, y)
:= \frac{1}{|G|}\sum_{g\in G} \code(g\cdot x,\dmV|\mV,g\cdot y),
\end{equation}
whose operational interpretation is given in Fig \ref{fig:sym2},
and which is also non-signalling. The value of $\bar{\code}(x,\dmV|\mV, y)$
depends only on $G(x,y)$, that is, the orbit of $(x,y)$ under the joint
action of $G$, and since
\begin{align}
	&\Pr(\dmRV\is \dmV|\mRV\is \mV,\bar{\code},\chan,\src_M)\\
	= & \frac{1}{|G|}\sum_{g\in G}
	\sum_{x \in \inA}\sum_{y\in \outA}
	\code(x,\dmV|\mV,y) \mc{E}(g \ga y | g \ga x) \\
	= & \frac{1}{|G|}\sum_{g\in G}
	\sum_{x \in \inA}\sum_{y\in \outA}
	\code(g^{-1}\ga x,\dmV|\mV,g^{-1}\ga y) \mc{E}(y | x)\\
	= & \Pr(\dmRV\is \dmV|\mRV\is \mV,\code,\chan,\src_M),
\end{align}
the optimisations for $\pe{\NS}{M}{\mc{E}}$ and $\mM{\NS}{\epsilon}{\mc{E}}$
in Definition \ref{code-params} can be restricted
to codes with this symmetry.

Proposition (\ref{errorPrimal}) was obtained by showing that one can
already to NS codes of the form
\begin{equation}
	\code(x,\dmV|\mV,y) = \begin{cases}
							R_{xy} &\text{ if } \dmV = \mV,\\
							\frac{u_{x} - R_{xy}}{M-1} &\text{ if }
							\dmV \neq \mV,
						\end{cases}
\end{equation}
without increasing the optimal error probability.
Applying the symmetrisation (\ref{sym2})
to this expression, $R_{xy}$ and $p(x)$ will only depend on
$G(x,y)$ and $Gx$, respectively. 

An example where symmetry can be used to great effect is where
$\mc{E}^n$ (with input alphabet is $\inA^n$ and output alphabet $\outA^n$)
is invariant under the actions of the symmetric group $S^n$
that permutes the symbols in the
input and output strings. This is true for any DMC, for example.

Following \cite{CK,MoT}, the \emph{joint type} of a pair of sequences
$\mathbf{x} = x_1 \ldots x_n \in \inA^n$ and
$\mathbf{y} = y_1 \ldots y_n \in \outA$,
is the distribution $P_{\mathbf{x},\mathbf{y}}$ on $\inA \times \outA$
defined by
$
	n P_{\mathbf{x},\mathbf{y}}(a,b) = N(a,b|\mathbf{x},\mathbf{y})
$
where $N(a,b|\mathbf{x},\mathbf{y})$ is the number of values of $i$ for
which $(x_i,y_i) = (a,b)$. $\mathcal{P}_n(\inA \times \outA)$ denotes
the set of all such joint types.
Likewise, the \emph{type} of a sequence
$\mathbf{x} \in \inA^n$ is the distribution $P_{\mathbf{x}}$ on $\inA$ with
$N(a|\mathbf{x}) = n P_{\mathbf{x}}(a)$ and $\mathcal{P}_n(\inA)$ is the
set of these.
Given a joint type
$\tau_{\inA\outA}$, the \emph{joint type class}
$\tc^n_{\tau_{\inA\outA}}$ is the set of all pairs of strings
$(\mathbf{x},\mathbf{y})$ with joint type $\tau_{\inA\outA}$.
Similarly, for a type $\tau_{\inA}$,
$\tc^n_{\tau_{\inA}} :=
\{ \mathbf{x} \in \inA^n : P_{\mathbf{x}} = \tau_{\inA} \}
$.

As is well known, the orbits of
$\inA^n \times \outA^n$ under the joint action of the symmetric
group described above,
are precisely the joint type classes $\tc^n_{\tau_{\inA\outA}}$,
for each joint type in
$\jointtypes$, and $\mc{E}^n(\mb{y}|\mb{x})$ is a function only of
the joint type of $\mb{x}$ and $\mb{y}$:
\begin{equation}
	\mc{E}^n(\mb{y}|\mb{x}) = \mc{E}^n(P_{\mb{x},\mb{y}}).
\end{equation}
For a DMC, and joint type $\tau_{\inA\outA} \in \jointtypes$,
\begin{equation}
	\mc{E}^n(P_{\mb{x},\mb{y}})
	= \prod_{a \in \inA,b \in \outA} \chan(b|a)^{n \tau_{\inA\outA}(a,b)}.
\end{equation}

Therefore, in the primal formulation of $\pe{\NS}{M}{\mc{E}^n}$
(Proposition \ref{errorPrimal})
one can take $R_{xy} = R(P_{\mb{x},\mb{y}})$
and $p(x) = p(P_{\mb{x}})$ for all $\mb{x},\mb{y}$, and
replace the sums over the input and output strings with
sums over joint types (or types)
which incorporate the correct multiplicity factors.

The objective function in (\ref{pObj}) becomes
\begin{equation}
	\sum_{\jtVar \in \jointtypes}
	|\tc^n_{\jtVar}| R(\jtVar) \mc{E}(\jtVar),
\end{equation}
where
$
	|\tc^n_{\jtVar}|
	= n!/(\prod_{a \in \inA, b \in \outA}
	(n\jtVar(a,b))!).
$
Similarly, the normalisation of $u$ (\ref{RuNorm}) becomes
\begin{equation}
\sum_{\sigma \in \mathcal{P}_n(\inA)} | \tc^n_{\sigma} |  u(\sigma_{\inA}) = 1.
\end{equation}
where $| \tc^n_{\sigma} | = n!/(\prod_{a \in \inA}
(n\sigma(a))!)$.
In (\ref{RMeq}) there is a constraint on a sum over $\inA^n$ for each
output string in $\outA^n$.
The number of pairs $(\mathbf{x},\mathbf{y})$ with joint type equal to
$\jtVar$ for fixed $\mathbf{y}$, depends only on $P_{\mathbf{y}}$, and is
equal to
\begin{equation}
	m(\tau;P_{\mathbf{y}}) := \begin{cases}
		\displaystyle\prod_{b \in \outA}
		\frac{(n\tau_{\outA}(b))! }
		{\prod_{a \in \inA} (n\jtVar(a,b))!}
			&\text { if } \tau_{\outA} = P_{\mathbf{y}},\\
		0 & \text { otherwise}.
	\end{cases}
\end{equation}
(Note that if the joint type of
$(\mathbf{x}, \mathbf{y})$
is $\tau_{\inA\outA}$,
then
the marginal distribution $\tau_{\inA}$ is the type of $\mathbf{x}$,
and $\tau_{\outA}$ is the type of $\mathbf{y}$.)
Therefore, (\ref{RMeq}) can be replaced by
\begin{equation}
	\forall \sigma_{\outA} \in \mathcal{P}_n (\outA) :
	\sum_{\jtVar \in \jointtypes} m(\jtVar;\sigma_{\outA}) R(\jtVar) \leq 1/M.
\end{equation}
The remaining constraints are equivalent to
\begin{align}
	\forall \jtVar \in \jointtypes &:
	0 \leq R(\tau_{\inA\outA}) \leq p(\tau_\inA).
\end{align}

Since the number of (joint) types is polynomial in $n$ \cite{CK,MoT}
the number of variables and constraints in the simplified LP given above
is polynomial in $n$, and this is also true of dual of this program.
The linear programs derived for $\mM{\NS}{\epsilon}{\mc{E}^n}$
can be simplified similarly.

\section{Assisted zero-error capacities of discrete memoryless channels
with zero dispersion}\label{dispersion}
As shown in \cite{CLMW2}, for any DMC, it follows from
Proposition \ref{zero-error} and the
multiplicitivity of the fractional packing number, that
\begin{align}
	C_0^{\NS}(\bm{\mc{E}}) = \log \alpha(H(\mc{E}))^{\ast}.
\end{align}
\def\sml{\mathcal{L}}

A \emph{simulation} $\sml$ of size $\kappa$ for the channel use
$\mc{E}$ of size $\kappa$
consists of an encoder which takes an input $X$ from $\inA$ and produces
a message $J$ in $\{ 1,\ldots,\kappa \}$,
and a decoder which takes a message $\hat{J}$ in $\{ 1,\ldots,\kappa\}$
and produces
an output $Y$ from $\outA$.
The $\sml$ determines the probabilities
\begin{equation}
	\sml(j,y|x,\hat{j}) :=
\Pr(J=j,Y=y|X=x,\hat{J}=\hat{j},\sml).
\end{equation}
We assume that the message is perfectly transmitted from the
encoder to the decoder, so $\hat{J} = J$.
The simulation is \emph{exact} if
\begin{align}
	\Pr(Y = y|X = x, \mathcal{L}) = \sum_{j=1}^{\kappa}\sml(j,y|x,j)
	= \mc{E}(y|x).
\end{align}
A non-signalling (NS) simulation is one where
\begin{align}
	\P{Y\is y|X\is x,\hat{J}\is \hat{j},\sml}
	=& \P{Y\is y|\hat{J}\is \hat{j},\sml},\\
	\P{J\is j|X \is x,\hat{J}\is \hat{j},\sml}
	=& \P{J\is j|X \is x,\sml}.
\end{align}
$\kappa^{\NS}_0(\mc{E})$ denotes
the minimum size of an exact NS simulation
of $\mc{E}$, and
\begin{equation}
	K_0^{\NS}(\bm{\mc{E}})
	:= \lim_{n\to\infty} \frac{1}{n} \log \kappa^{\NS}_0(\mc{E}^n)
\end{equation}
is the (asymptotic) \emph{exact simulation cost} of $\bm{\mc{E}}$.
In \cite{CLMW2} it was shown that, for any DMC,
\begin{equation}
	K_0^{\NS}(\bm{\mc{E}})
	= \log \sum_{y \in \outA} \max_{x \in \inA} \mc{E}(y|x).
\end{equation}

In what follows, $\bm{\mc{E}}$ is omitted as an argument, since it refers to
some fixed channel. For any discrete channel,
\begin{equation}\label{sandwich}
	C_0^{}
	\leq C_0^{\SE}
	\leq C_0^{\NS}
	\leq C
	\leq K_0^{\NS}.
\end{equation}

From now on, let $\bm{\mc{E}}$ be a DMC with $\mc{E}^n = \mc{E}^{\ox n}$.
Proposition 26 of \cite{CLMW2} shows that, given any requirement on which
transition probabilities in $\mc{E}$ must be zero, it is possible to find
an $\mc{E}$ that satisfies that requirement and has all three quantities
in (\ref{sandwich}) equal. In \cite{LMMOR} it was shown that there are
DMCs where even the
\emph{entanglement-assisted} zero-error capacity $C_0^{\SE}$
reaches $C$
(and with a block length one entanglement-assisted code) despite the
unassisted zero-error capacity $C_0^{}$ being strictly smaller.

For a DMC,
$V$ is the minimum variance of the information density of the
channel for the joint distribution induced by any capacity achieving input
distribution for a single channel use. The information density is
\begin{equation}\label{infd}
	i(x;y) = \log \frac{\mc{E}(y|x)}{q(y)},
\end{equation}
where $q \in \mathcal{P}(\outA)$
is the output distribution. The capacity is the expectation of
the information density. Therefore, $V = 0$ iff there exists a
capacity achieving input distribution $p \in \mathcal{P}(\inA)$
(with induced output distribution $q$)
such that (\ref{infd}) is equal to $C$ when $\mc{E}(y|x)p(x)$
is non-zero, i.e. if and only if, for $x$ s.t. $p(x) > 0$
\begin{equation}\label{V0}
	\mc{E}(y|x) = \lceil \mc{E}(y|x) \rceil q(y) 2^{C}.
\end{equation}
If $V$ is zero, then the $\sqrt{n}$ term vanishes in the asymptotic
expansion. In this sense, a channel with zero dispersion admits qualitatively
more efficient codes (in terms of approaching capacity with increasing
block length)
than a channel with positive variance does. It turns out that a channel has
zero dispersion if and only if its capacity can be achieved with
zero-error by NS codes.
\begin{theorem}
	For a DMC $\bm{\mc{E}}$ the three conditions
	\begin{enumerate}
		\item $C^{\rm{NS}}_0(\bm{\mc{E}}) = C(\bm{\mc{E}})$,
		\item $K^{\rm{NS}}_0(\bm{\mc{E}}) = C(\bm{\mc{E}})$,
		\item $V(\bm{\mc{E}}) = 0$.
	\end{enumerate}
	are equivalent.
\end{theorem}
\begin{proof}
The following propositions show
that (3) implies (1) and (2); that (1) implies (3); and that (2)
implies (3).
\end{proof}

\begin{proposition}
	If $V = 0$ then
	$C_0^{\NS}$ and
	$K_0^{\NS}$ are
	both equal to $C$.
\end{proposition}
\begin{proof}
	We show that if $V(\mc{E}) = 0$ then the opposite inequalities to those
	in (\ref{sandwich}) also hold. Using (\ref{V0}),
	\begin{align}
		K^{\rm{NS}}_0
		=&\log \sum_{y\in \outA} \max_x \mc{E}(y|x)\\
		=&\log \sum_{y\in \outA} \max_x \lceil \mc{E}(y|x) \rceil
		q(y) 2^{C} \leq C.
	\end{align}
	For the other part, when $q(y)$ is non-zero
	\begin{align}
		\sum_{x\in \inA} \lceil \mc{E}(y|x) \rceil p(x) 2^{C}
		= \sum_{x\in \inA} \frac{\mc{E}(y|x)}{q(y)} p(x) = 1,
	\end{align}
	and when $q(y)$ is zero we must have $\lceil \mc{E}(y|x) \rceil p(x) = 0$
	for all $x \in \inA$ and
	\begin{align}
		\sum_{x\in \inA} \lceil \mc{E}(y|x) \rceil p(x) 2^{C} = 0.
	\end{align}
	Therefore $p(x) 2^{C}$ is a fractional packing, and
	$C^{\NS} \geq C$.
\end{proof}

\begin{definition}
	\begin{equation}
		\alpha^{\ast}(\mc{E},p) := \max \{ \alpha : \forall y~\sum_{x}
		\lceil \mc{E}(y|x) \rceil \alpha p(x) \leq 1 \},
	\end{equation}
	which is equivalent to
	\begin{equation}
		\alpha^{\ast}(\mc{E},p) = \frac{1}{\max_y \sum_{x\in \inA} 
		\lceil \mc{E}(y|x) \rceil p(x) }.
	\end{equation}
\end{definition}
Clearly the fractional packing number is given by
$\alpha^{\ast}(\mc{E}) = \max_p \alpha^{\ast}(\mc{E},p)$,
where the maximum is over probability distributions $p$ on the input alphabet.

\begin{lemma}\label{astarI}
	Let $I(\mc{E},p)$ denote the mutual information between channel input and
	output when the input has probability mass function $p$.
	\begin{equation}
		I(\mc{E},p) \geq \log \alpha^{\ast}(\mc{E},p)
	\end{equation}
\end{lemma}
\begin{proof}
	Let $q(y) = \sum_{x \in \inA} p(x) \mc{E}(y|x)$.
	\begin{align}
		\log \alpha^{\ast}(\mc{E},p) 
		&= - \max_{y \in \outA} \log \sum_{x\in \inA}
		\lceil \mc{E}(y|x) \rceil p(x)\\
		&\leq - \sum_{y \in \outA} q(y) \log
		\sum_{x\in \inA} \lceil \mc{E}(y|x) \rceil p(x)\\
		&= - \sum_{y\in \outA}\sum_{x\in \inA} \mc{E}(y|x) p(x) \log
		\sum_{x'} \lceil \mc{E}(y|x') \rceil p(x')
	\end{align}
	Subtracting $I(\mc{E},p)$ from this last expression one obtains
	\begin{equation}
		\sum_{y\in \outA}\sum_{x\in \inA}
		\mc{E}(y|x) p(x) \log \frac{\sum_{x''}
		\mc{E}(y|x'')p(x'')}
		{\mc{E}(y|x) \sum_{x'} \lceil \mc{E}(y|x') \rceil p(x')}
	\end{equation}
	which is never larger than zero because,
	using $\log x \leq (x - 1)/(\ln 2)$,
	\begin{align}
		& \sum_{x : \mc{E}(y|x) > 0} \mc{E}(y|x) p(x)
		\log \frac{\sum_{x''} \mc{E}(y|x'')p(x'')}
		{\mc{E}(y|x) \sum_{x'} \lceil \mc{E}(y|x') \rceil p(x')}\\
		\leq & \sum_{x : \mc{E}(y|x) > 0} \frac{\mc{E}(y|x) p(x)}{\ln 2}
		\left(\frac{\sum_{x''} \mc{E}(y|x'')p(x'')}
		{\mc{E}(y|x) \sum_{x'}
		\lceil \mc{E}(y|x') \rceil p(x')} - 1\right)\\
		\begin{split} =&
		\bigg(\frac{\sum_{x''} \mc{E}(y|x'')p(x'')}
		{\sum_{x'} \lceil \mc{E}(y|x') \rceil p(x')}
		\sum_{x} \lceil \mc{E}(y|x) \rceil p(x)\\
		&- \sum_{x\in \inA} \mc{E}(y|x) p(x)\bigg)/(\ln 2)\end{split}\\
		= & 0.
	\end{align}
\end{proof}

\begin{proposition}
	If $C^{\rm{NS}}_0 = C$ then $V = 0$.
\end{proposition}
\begin{proof}
Suppose that $\mc{E}(y|x)$ is a channel with
$C^{\rm{NS}}_0 = C$.
Let $w: A \to [0,1]$ be any optimal fractional packing for the channel and
let $\alpha^{\ast}$ be the fractional packing number. By the preceding lemma,
$p(x) = w(x)/\alpha^{\ast}$ defines a capacity achieving input probability
mass function for the channel. Let $q$ be the corresponding output probability
mass function. It was shown in \cite{Shannon} that if $p$ is capacity
achieving then
\begin{equation}
	D(\mc{E}(\cdot|x) || q) \begin{cases}= C \text{ when } p(x) > 0,\\
	\leq C \text{ when } p(x) = 0.
	\end{cases}
\end{equation}
Since, $C = \log \alpha^{\ast}$ by assumption, these conditions imply
that, for all $x \in \inA$,
\begin{align}
	0 \leq \log \alpha^{\ast} - D(\mc{E}(\cdot|x) || q)
	= \sum_{y\in \outA} \mc{E}(y|x) \log \frac
	{ q(y) \alpha^{\ast}}
	{\mc{E}(y|x) }
\end{align}
and, using $\log x \leq (x - 1)/(\ln 2)$ again,
\begin{align}
	0 \leq & \sum_{y \in \outA} \mc{E}(y|x) \lceil \mc{E}(y|x) \rceil
	\left( \frac{q(y) \alpha^{\ast}}
	{\mc{E}(y|x) } - 1 \right)\\
	= & \sum_{y\in \outA}
	\alpha^{\ast} q(y) \lceil \mc{E}(y|x) \rceil - \sum_{y} \mc{E}(y|x)\\
	= & \sum_{y\in \outA}
	\alpha^{\ast} q(y) \lceil \mc{E}(y|x) \rceil - 1.
\end{align}
Therefore, $v_y := \alpha^{\ast} q(y)$ is a fractional covering
for the channel hypergraph, and it is optimal. Furthermore,
the complementary slackness condition demands that
when $p(x) > 0$ the corresponding inequality must be
saturated. Therefore, when $\mc{E}(y|x)p(x) > 0$ we must have
$\frac
{q(y) \alpha^{\ast}	}
{\mc{E}(y|x) } - 1 = 0$ or $\log \mc{E}(y|x)/q(y) = \log \alpha^{\ast}$
so the variance of the information density is zero for this capacity
achieving distribution.
\end{proof}
\begin{proposition}
	If $K_0^{\NS} = C$ then $V = 0$.
\end{proposition}
\begin{proof}
	Let $p$ be a capacity achieving probability mass function.
	\begin{align}
		C = & \sum_{x\in \inA}\sum_{y\in \outA}
		\mc{E}(y|x) p(x) \log \frac{\mc{E}(y|x)}{\sum_{x'}\mc{E}(y|x')p(x')}\\
		\leq & \log \sum_{x\in \inA}\sum_{y\in \outA}
		\mc{E}(y|x) p(x) \frac{\mc{E}(y|x)}{\sum_{x'}\mc{E}(y|x')p(x')}
		\label{jensen}\\
		\leq & \log \sum_{y\in \outA} \max_{x'' \in \inA}
		\mc{E}(y|x'')
		\frac{\sum_{x} p(x) \mc{E}(y|x)}
		{\sum_{x'}\mc{E}(y|x')p(x')}\label{max}\\
		= & \log \sum_{y\in \outA} \max_{x\in \inA} \mc{E}(y|x)\\
		= & K_0^{\NS}.
	\end{align}
	For equality to hold, Jensen's inequality (\ref{jensen})
	must be saturated. This happens if and only if
	\begin{equation}
		\frac{\mc{E}(y|x)}{\sum_{x'\in \inA}\mc{E}(y|x')p(x')}
	= C
	\end{equation}
	for all $x,y$ such that $\mc{E}(y|x) p(x) > 0$,
	which is equivalent to $V = 0$.
\end{proof}

\section{Conclusion}
It was shown that maximum size of non-signalling code
with a given error probability is given by the integer part
of the solution to a linear program,
and that this is equal to the converse bound of Polyanskiy, Poor and
Verd\'{u}~\cite{PPV},
thus giving an alternative proof of that result.
When $n$ uses of the channel are symmetric under simultaneous permutations
of the input and output strings, the LP can be simplified to one with
$\text{poly}(n)$ variables and constraints.

It was also proven that the capacity of a DMC is achieved with
zero-error by NS codes, if and only if the channel has zero dispersion,
and therefore already admits especially efficient classical codes.

It would be interesting to see if the dual linear programming formulation of
the converse given in this paper can help in extending the finite block length
results given in \cite{PPV}. The technique of using non-signalling assistance
to obtain linear program converses for classical coding protocols extends
naturally to multi-terminal situations like broadcast or multiple access
channels, and may prove useful in this context.

\section*{Acknowledgments}
I would like to thank Andreas Winter, Toby Cubitt and Debbie Leung
for useful discussions.


\end{document}